\documentclass[a4paper,UKenglish]{lncs-template/llncs/llncs}

\newcommand{\PMSSCWithAppendix}{}
\usepackage[T1]{fontenc}
\usepackage[utf8]{inputenc}
\usepackage[UKenglish]{babel}
\AtBeginDocument{%
  \renewcommand{\abstractname}{Abstract.}%
  \renewcommand{\figurename}{Fig.}%
}
\usepackage{amsmath,amssymb,mathtools}
\usepackage{xspace}
\usepackage{xcolor}
\usepackage[ruled,vlined,linesnumbered]{algorithm2e}
\SetArgSty{textnormal}
\usepackage{tikz}
\usepackage{float}
\usetikzlibrary{decorations.pathreplacing,positioning,calc,arrows.meta}
\usepackage[hidelinks]{hyperref}
\usepackage[noabbrev,nameinlink]{cleveref}

\newcommand{\br}[1]{\mathopen{}\left( #1 \right)}
\newcommand{\brc}[1]{\mathopen{}\left\{ #1 \right\}}
\newcommand{\spr}[1]{\mathopen{}\left| #1 \right|}
\newcommand{\fl}[1]{\mathopen{}\left\lfloor #1 \right\rfloor}
\newcommand{\cl}[1]{\mathopen{}\left\lceil #1 \right\rceil}
\newcommand{\angl}[1]{\mathopen{}\langle #1 \rangle}

\newcommand{\E}[1]{\mathbb{E}\left[#1\right]}

\newcommand{\cov}{\operatorname{cov}}
\newcommand{\supp}{\operatorname{supp}}
\newcommand{\ct}{\operatorname{ct}}
\newcommand{\ft}{\operatorname{ft}}
\newcommand{\lt}{\operatorname{lt}}

\newcommand{\OPT}{\text{OPT}}
\newcommand{\COST}{\operatorname{cost}}
\newcommand{\Cmax}{C_{\max}}

\newcommand{\argmax}{\mathopen{}\operatorname*{arg\,max}}

\newcommand{\cL}{\mathcal{L}}
\newcommand{\cS}{\mathcal{S}}

\newcommand{\cC}{\mathcal{C}}
\newcommand{\cF}{\mathcal{F}}
\newcommand{\cU}{\mathcal{U}}
\newcommand{\cR}{\mathcal{R}}
\newcommand{\cI}{\mathcal{I}}

\newcommand{\cA}{\mathcal{A}}

\newcommand{\cO}{\mathcal{O}}
\newcommand{\cX}{\mathcal{X}}
\newcommand{\cK}{\mathcal{K}}
\newcommand{\cP}{\mathcal{P}}

\newcommand{\ProblemPMSSC}{\textsc{PMSSC}}

\newcommand{\ProblemPDS}{\textsc{PDS}}
\AddToHook{env/proof/end}{\leavevmode\qed}

\newcommand{\PMSSCCameraReady}{}
\ifdefined\PMSSCWithAppendix\else
  \newcommand{\PMSSCProceedings}{}
\fi

\title{Min-Sum Set Cover on Parallel Machines}
\titlerunning{Parallel Min-Sum Set Cover}
\author{Michał Szyfelbein\orcidID{0009-0009-9894-9671}\thanks{Corresponding author.}}
\authorrunning{M. Szyfelbein}
\institute{Gdańsk University of Technology, Gdańsk, Poland\\
\email{michal.szyfelbein@pg.edu.pl}}

\AddToHook{env/abstract/end}{%
}

\begin{document}
\pagestyle{empty}
\maketitle
\thispagestyle{empty}

\begin{abstract}
We consider a generalization of the \textsc{Min-Sum Set Cover} to the setup with $m$ set-sequences, or in scheduling terminology, $m$ parallel machines. We call this problem \textsc{Parallel Min-Sum Set Cover}. To obtain approximation algorithms for its numerous variants we use a crucial sub-problem called \textsc{Parallel Densest Subfamily}. We prove that an $\alpha$-approximation algorithm for this task gives a $4\cdot\alpha$-approximation for the \textsc{Parallel Min-Sum Set Cover}, which yields $\frac{4\cdot e}{e-1}+\epsilon$ and $4\cdot \br{\frac{e}{e-1}}^2+\epsilon$-approximation ratios for identical and unrelated machines, respectively.
To obtain the latter result we give a new $\br{\frac{e}{e-1}}^2+\epsilon$-approximation algorithm for the \textsc{Maximum Coverage Multiple Knapsacks} problem which is of independent interest.
If the sets are precedence-constrained, for unit cost sets we give an $\mathcal{O}(k^{2/3})$ approximation ($k$ is the number of sets). For out-forest precedence constraints, a reduction to \textsc{Group Steiner Orienteering} improves the bound to $\mathcal{O}(\log k)$. With $n$ universe elements, an $\mathcal{O}(\log^{1-\epsilon} n)$ approximation for any fixed $\epsilon>0$ would imply $NP\subseteq ZTIME(n^{\mathcal{O}(\text{poly}(\log n))})$.
\end{abstract}

\section{Introduction}
\label{sec:introduction}

Consider the following scenario. A bus infrastructure is to be provided for a set of $n$ people, represented by a universe $\cU$. The infrastructure plan consists of a collection $\cS$ of $k$ potential bus stops, where each set in $\cS$ represents the people for whom the corresponding stop would provide access to public transportation. The goal is to construct a collection of stops so that each person is provided with the access. We have $m$ machines at our disposal, representing teams of workers, each of which can construct at most one stop at a time. The time required to construct a given stop may vary both across locations and between teams. The objective is to find an $m$-machine schedule of sets from $\cS$ that covers all elements of $\cU$ while minimizing the sum of the covering times of the elements. Equivalently, we seek to minimize the average time a citizen has to wait before gaining access to public transportation. We call this problem \textsc{Parallel Min-Sum Set Cover}.
See Figure \ref{fig:pmssc-instance-and-schedule} for an example.
\begin{figure}[t]
\centering
\colorlet{sone}{blue!31}
\colorlet{stwo}{cyan!34}
\colorlet{sthree}{teal!36}
\colorlet{sfour}{violet!30}
\colorlet{sfive}{orange!42}
\colorlet{ssix}{red!34}
\colorlet{sseven}{green!38}
\colorlet{seight}{yellow!46}
\colorlet{snine}{magenta!30}
\colorlet{sten}{brown!38}

\ifdefined\PMSSCCameraReady
  \def\PMSSCIncidenceX{0.51cm}
  \def\PMSSCScheduleX{1cm}
  \tikzset{example note/.style={text width=0.90\linewidth,align=left},
           example title/.style={anchor=south west}}
\else
  \def\PMSSCIncidenceX{0.58cm}
  \def\PMSSCScheduleX{1.1cm}
  \tikzset{example note/.style={}, example title/.style={}}
\fi

\begin{tikzpicture}[
    x=\PMSSCIncidenceX,
    y=0.52cm,
    setnode/.style={draw, rounded corners=2pt, minimum width=1.1cm, minimum height=0.62cm, font=\scriptsize, inner sep=1.1pt},
    unode/.style={circle, draw, fill=green!8, minimum size=0.41cm, inner sep=0pt, font=\tiny}
]
    \node[setnode, fill=sone]   (s1)  at (1.6,8.0) {$S_1(1)$};
    \node[setnode, fill=stwo]   (s2)  at (3.8,8.0) {$S_2(2)$};
    \node[setnode, fill=sthree] (s3)  at (6.0,8.0) {$S_3(3)$};
    \node[setnode, fill=sfour]  (s4)  at (8.2,8.0) {$S_4(2)$};
    \node[setnode, fill=sfive]  (s5)  at (10.4,8.0) {$S_5(4)$};
    \node[setnode, fill=ssix]   (s6)  at (12.6,8.0) {$S_6(5)$};
    \node[setnode, fill=sseven] (s7)  at (14.8,8.0) {$S_7(1)$};
    \node[setnode, fill=seight] (s8)  at (17.0,8.0) {$S_8(3)$};
    \node[setnode, fill=snine]  (s9)  at (19.2,8.0) {$S_9(2)$};
    \node[setnode, fill=sten]   (s10) at (21.4,8.0) {$S_{10}(4)$};

    \foreach \i in {1,...,20} {
        \node[unode] (u\i) at (\i+1.25,2) {$u_{\i}$};
    }

    \node[font=\tiny, anchor=north] at (2.25,1.66) {};
    \node[font=\tiny, anchor=north] at (3.25,1.66) {};
    \node[font=\tiny, anchor=north] at (4.25,1.66) {};
    \node[font=\tiny, anchor=north] at (5.25,1.66) {};
    \node[font=\tiny, anchor=north] at (6.25,1.66) {};
    \node[font=\tiny, anchor=north] at (7.25,1.66) {};
    \node[font=\tiny, anchor=north] at (8.25,1.66) {};
    \node[font=\tiny, anchor=north] at (9.25,1.66) {};
    \node[font=\tiny, anchor=north] at (10.25,1.66) {};
    \node[font=\tiny, anchor=north] at (11.25,1.66) {};
    \node[font=\tiny, anchor=north] at (12.25,1.66) {};
    \node[font=\tiny, anchor=north] at (13.25,1.66) {};
    \node[font=\tiny, anchor=north] at (14.25,1.66) {};
    \node[font=\tiny, anchor=north] at (15.25,1.66) {};
    \node[font=\tiny, anchor=north] at (16.25,1.66) {};
    \node[font=\tiny, anchor=north] at (17.25,1.66) {};
    \node[font=\tiny, anchor=north] at (18.25,1.66) {};
    \node[font=\tiny, anchor=north] at (19.25,1.66) {};
    \node[font=\tiny, anchor=north] at (20.25,1.66) {};
    \node[font=\tiny, anchor=north] at (21.25,1.66) {};

    \foreach \u in {1,2}              {\draw[sone!85!black, line width=0.7pt]   (s1) -- (u\u);}       
    \foreach \u in {3,5,7}            {\draw[stwo!85!black, line width=0.7pt]   (s2) -- (u\u);}       
    \foreach \u in {6,8,10}           {\draw[sthree!85!black, line width=0.7pt] (s3) -- (u\u);}       
    \foreach \u in {4,9,11,13}        {\draw[sfour!85!black, line width=0.7pt]  (s4) -- (u\u);}       
    \foreach \u in {5,12,15,16}       {\draw[sfive!85!black, line width=0.7pt]  (s5) -- (u\u);}       
    \foreach \u in {8,10,14,18,20}    {\draw[ssix!85!black, line width=0.7pt]   (s6) -- (u\u);}       
    \foreach \u in {11,17}            {\draw[sseven!85!black, line width=0.7pt] (s7) -- (u\u);}       
    \foreach \u in {7,16,19}          {\draw[seight!85!black, line width=0.7pt] (s8) -- (u\u);}       
    \foreach \u in {14,20}            {\draw[snine!85!black, line width=0.7pt]  (s9) -- (u\u);}       
    \foreach \u in {9,12,18,19}       {\draw[sten!85!black, line width=0.7pt]   (s10) -- (u\u);}      

    \node[example note, font=\small, anchor=west, example title] at (0.5,9.1) {Example identical machine PMSSC instance: sets with costs at the top, universe at the bottom.};
\end{tikzpicture}

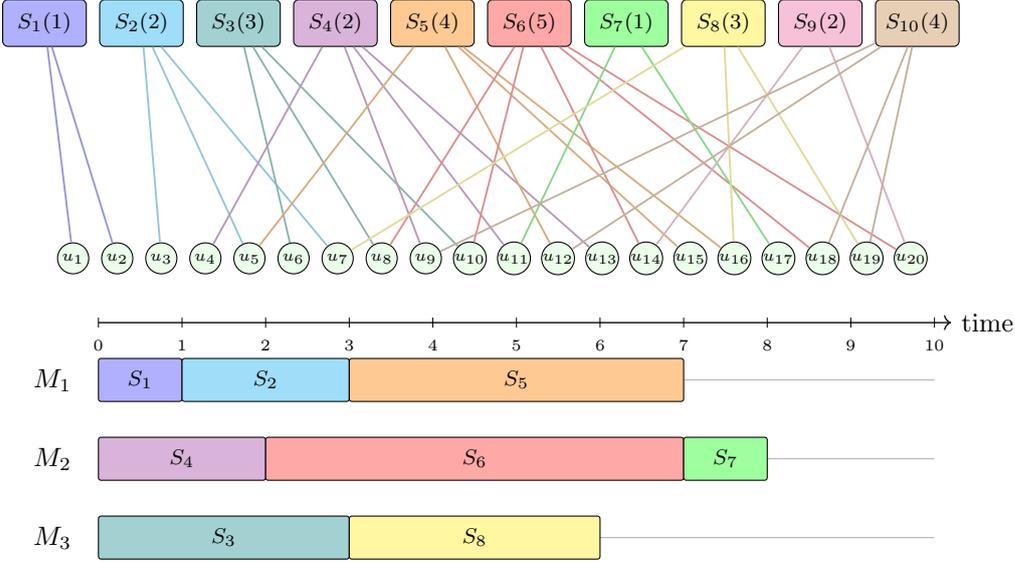
\ifdefined\PMSSCCameraReady\else
  \vspace{0.35cm}
\fi

\begin{tikzpicture}[
    x=\PMSSCScheduleX,
    y=0.95cm,
    lane/.style={line width=0.45pt},
    task/.style={draw, rounded corners=1pt, minimum height=0.58cm, font=\scriptsize, align=center},
    mlabel/.style={font=\small, anchor=east}
]
    \draw[->, line width=0.5pt] (0,0.6) -- (10.2,0.6) node[anchor=west, font=\small] {time};
    \foreach \t in {0,...,10} {
        \draw[line width=0.25pt] (\t,0.53) -- (\t,0.67);
        \node[font=\tiny, anchor=north] at (\t,0.5) {\t};
    }

    \node[mlabel] at (-0.2,-0.2) {$M_1$};
    \node[mlabel] at (-0.2,-1.3) {$M_2$};
    \node[mlabel] at (-0.2,-2.4) {$M_3$};
    \draw[lane, gray!55] (0,-0.2) -- (10,-0.2);
    \draw[lane, gray!55] (0,-1.3) -- (10,-1.3);
    \draw[lane, gray!55] (0,-2.4) -- (10,-2.4);

    \draw[task, fill=sone]   (0,-0.5) rectangle (1,0.1) node[pos=.5] {$S_1$};
    \draw[task, fill=stwo]   (1,-0.5) rectangle (3,0.1) node[pos=.5] {$S_2$};
    \draw[task, fill=sfive]  (3,-0.5) rectangle (7,0.1) node[pos=.5] {$S_5$};

    \draw[task, fill=sfour]  (0,-1.6) rectangle (2,-1.0) node[pos=.5] {$S_4$};
    \draw[task, fill=ssix]   (2,-1.6) rectangle (7,-1.0) node[pos=.5] {$S_6$};
    \draw[task, fill=sseven] (7,-1.6) rectangle (8,-1.0) node[pos=.5] {$S_7$};

    \draw[task, fill=sthree] (0,-2.7) rectangle (3,-2.1) node[pos=.5] {$S_3$};
    \draw[task, fill=seight] (3,-2.7) rectangle (6,-2.1) node[pos=.5] {$S_8$};

    \node[example note, font=\small, anchor=west] at (0,-3.35) {Example schedule on 3 machines of cost $83$, note that $S_9$ and $S_{10}$ are not scheduled.};
\end{tikzpicture}
\caption{An example instance and a feasible schedule for \textsc{Parallel Min-Sum Set Cover}.}
\label{fig:pmssc-instance-and-schedule}
\end{figure}

\subsection{Our results and techniques}
\label{sec:our-results}

We present a series of approximation algorithms for the \textsc{Parallel Min-Sum Set Cover} (\textsc{PMSSC}) problem. Crucial for our analysis is the choice of an appropriate subproblem, called the \textsc{Parallel Densest Subfamily} (PDS). In Section \ref{sec:pmssc-scheme}, we show how to greedily use an $\alpha$-approximation algorithm for PDS to obtain a $4\cdot\alpha$-approximation algorithm for PMSSC. Our argument generalizes a particularly elegant analysis of the greedy algorithm for the \textsc{Min-Sum Set Cover} introduced by Feige et.~al.~\cite{ApproximatingMinSumSetCover}, and utilized in few subsequent works \cite{Precedence-ConstrainedMinSumSetCover,AGeneralFrameworkForApproximatingMinSumOrderingProblems}. 
 In order to solve PDS, we further reduce it to the well-known \textsc{Maximum Coverage Multiple Knapsacks} problem yielding $\frac{e}{e-1}+\epsilon$ and $\br{\frac{e}{e-1}}^2+\epsilon$-approximation ratios for PDS on identical and unrelated machines, respectively. The latter result is based on a new configuration-LP-based algorithm for a general MCMK with approximation of $\br{\frac{e}{e-1}}^2+\epsilon$, improving over $\frac{2\cdot e}{e-1}$, implicitly given in \cite{MaximumCoverageWithClusterConstraints}. 
 Consequently, we get $\frac{4\cdot e}{e-1}+\epsilon< 6.33$ and $4\cdot\br{\frac{e}{e-1}}^2+\epsilon< 10.02$-approximation ratios for \textsc{PMSSC} on identical and unrelated machines as well.

In Section \ref{sec:precedence}, we shift towards precedence-constrained instances. We show a greedy algorithm with an approximation ratio of $\cO\br{k^{2/3}}$ for the precedence-constrained analoge of PDS giving the same approximation for PCMSSC itself.
We also consider special case of out-forest precedence constraints, where each set has at most one direct predecessor. Using a reduction to the \textsc{Group Steiner Orienteering} problem, we get an $\cO\br{\log k}$-approximation. We also show that for this setup achieving an $\cO\br{\log^{1-\epsilon} n}$-approximation is impossible, unless $NP\subseteq ZTIME\br{n^{\text{poly}\br{\log n}}}$. For general constraints it is known that even case $m=1$ is inapproximable within a factor of $\cO\br{k^{1/12-\epsilon}}$ or $\cO\br{n^{1/6-\epsilon}}$ for any $\epsilon>0$, unless the \textsc{Planted Dense Subgraph} conjecture is false \cite{Precedence-ConstrainedMinSumSetCover}.
\subsection{Related Work}

The \textsc{Set Cover} is a cornerstone problem in the combinatorial optimization. Vazirani uses it as a prime example for various algorithmic techniques in his book \cite{vazirani2001approximation}. The greedy algorithm achieves an $H_n$-approximation guarantee, where $H_n$ is the $n$-th harmonic number, and this is tight under the assumption that $P\neq NP$~\cite{AnalyticalApproachToParallelRepetition}. If each element belongs to at most $f$ sets, then there exists a layering algorithm which achieves an $f$-approximation. The related problem called \textsc{Maximum Coverage} (\textsc{MC}) is to cover maximum number of elements using a given budget. The greedy algorithm for this problem achieves an $\frac{e}{e-1}$-approximation guarantee, and this is best possible as well (under $P\neq NP$ ) \cite{TheBudgetedMaximumCoverageProblem}. Important generalizations of \textsc{MC} include \textsc{Submodular Optimization}, where one maximizes a submodular function under additional constraints. Multiple variants of this setup have been studied starting with the work of Nemhauser et al. \cite{AnAnalysisofApproximationsForMaximizingSubmodularSetFunctionsI}.

For the \textsc{Min-Sum Set Cover}, introduced by Feige et.~al.~\cite{ApproximatingMinSumSetCover}, the greedy algorithm achieves $4$-approximation and this is best possible. Since then, their work has been extended towards various directions. Happach et.~al.~\cite{AGeneralFrameworkForApproximatingMinSumOrderingProblems} present a general framework for approximating min-sum ordering problems and show that greedy achieves a $4\cdot\alpha$-approximation guarantee, where $\alpha$ is the approximation ratio for the problem of finding the densest subset. McClintock et.~al.~\cite{Precedence-ConstrainedMinSumSetCover} consider a generalization of \textsc{Min-Sum Set Cover} where sets are precedence-constrained, and they show that a modified  greedy which uses the densest precedence-closed subfamily achieves a $4\cdot\alpha$-approximation guarantee, given an $\alpha=\sqrt{k}$-approximation algorithm for the latter problem. Moreover, the problem cannot be approximated within the factor of $\cO\br{k^{1/12-\epsilon}}$ nor $\cO\br{n^{1/6-\epsilon}}$ for any $\epsilon>0$ unless the \textsc{Planted Dense Subgraph} conjecture fails. These results have been further extended by Szyfelbein et.~al.~\cite{Precedence-ConstrainedDecisionTreesAndCoverings} to an $\cO^*\br{\sqrt{k}}$-approximation for the case where only a certain fraction of elements is required to be covered.

\textsc{Parallel Task Scheduling} is a well-studied task in the scheduling literature. The problem $P||C_{max}$ is strongly $NP$-hard and has an EPTAS running in $2^{\cO\br{1/\epsilon\cdot \log\br{1/\epsilon}\cdot\log\log\br{1/\epsilon}}}$ time \cite{LoadBalancing:TheLongRoadFromTheoryToPractice}. The problem $Q||C_{max}$ also admits an EPTAS, but running in ${2^{\cO\br{1/\epsilon\cdot\log^4\br{1/\epsilon}}}}$ time \cite{ClosingTheGapForMakespanSchedulingViaSparsificationTechniques}. However, $R||C_{max}$ cannot be approximated within a factor of $3/2-\epsilon$ for any $\epsilon>0$, unless $P=NP$ and the best known approximation algorithm achieves a factor of $2$ \cite{ApproximationAlgorithmsForSchedulingUnrelatedParallelMachines}. 
The problem $R||\sum C_j$ is polynomially solvable~\cite{SchedulingIndependentTasksToReduceMeanFinishingTime}, however the addition of weights ($R||\sum w_jC_j$) makes it strongly $NP$-hard \cite{ComplexityOfMachineSchedulingProblems}. For $Q||\sum w_jC_j$ a PTAS is known \cite{APTASForMinimizingWeightedCompletionTimeOnUniformlyRelatedMachines}, but the $R||\sum w_jC_j$ is APX-hard \cite{NonApproximabilityResultsForSchedulingProblemsWithMinsumCriteria} and the curently best known approximation is $1.488$~\cite{WeightedCompletionTimeMinimizationForUnrelatedMachinesViaIterativeFairContentionResolution}.
\subsection{Motivation and Applicatons}
The list of applications of \textsc{Parallel Min-Sum Set Cover} is beyond the scope of this work, hereby we showcase some of the more interesting examples:
\begin{enumerate}
    \item \textbf{Software testing:} A system containing a set of features is to be provided. We wish to verify whether these features work properly, and if not, identify the faulty components. A set of tests is at our disposal, each of which can verify some subset of features and a number of computers which can execute these tests. Moreover, we have a cost function associated with each test and each computer which denotes the amount of time required to execute that test on this computer. We want to schedule the tests on the computers in such a way that the average time required to verify a feature is minimized.
    \item \textbf{Blackout fixing:} We have a power grid (or any other large infrastructure) that serves many users. A massive failure occurs and the source of failure is required to be identified as soon as possible. Each diagnostic action (performed at a chosen location) can test several potentially overlapping components because failures may propagate through shared dependencies. With multiple repair teams, we want to schedule these diagnostic actions so that the average time to identify the faulty component is minimized.
    \item \textbf{Iterative feature release:} A software product is to be provided to customers. We have a set of features which to implement, and a set of teams of workers at our disposal, each of which can be assigned to implement a certain component consisting of a subset of features at a given time. We want to schedule the work of the teams in such a way that the average time of feature implementation is minimized.
\end{enumerate}


\section{Preliminaries}
\label{sec:preliminaries}
More formally, the input of the \textsc{Parallel Min-Sum Set Cover} ($\ProblemPMSSC$) is as follows. We are given a universe $\cU$ of $n$ elements, a collection $\cS$ of $k$ covering subsets of $\cU$ ($\bigcup_{S\in \cS} S = \cU$), $m$ machines and a cost function\footnote{The value $c\br{S,j}$ can also be called the \emph{processing time}. We use the term \emph{cost} since we needed to choose between the terminologies of covering and scheduling problems.} $c\colon \cS\times [m]\to \mathbb{N}$, where $[m]=\{1, 2, \dots, m\}$. We are also given a weight function\footnote{The weight function does not significantly affect our analysis but adds generality.} $w\colon \cU\to \mathbb{N}$. A \emph{schedule} $\sigma$ is a collection of $m$ sequences of sets from $\cS$. 
Let $\cC\subseteq \cS$ and let $\cU'\subseteq \cU$. The \emph{coverage} of $\cC$ in $\cU'$ is defined as $\cov\br{\cC, \cU'}=\bigcup_{S\in \cC} S\cap \cU'$. If $\cU'=\cU$, we simply write $\cov\br{\cC}$. For a sequence of sets $\pi=\angl{S_1, S_2, \dots, S_t}$ the coverage is defined analogously as $\cov\br{\pi, \cU'}=\bigcup_{j=1}^{t} S_j \cap \cU'$. Let $\sigma=\brc{\sigma\br{j}}_{j=1}^m$ be a collection of $m$ sequences of sets. Then, the coverage of $\sigma$ is defined as $
\cov\br{\sigma, \cU'}=\bigcup_{j=1}^{m} \cov\br{\sigma\br{j}, \cU'}.$
A~schedule $\sigma$ is a \emph{feasible solution} for $\ProblemPMSSC$ if $\cov\br{\sigma, \cU}=\cU$. Let $S\in \cS$ be some set scheduled in $\sigma$ on machine $j_S\in [m]$. Denote by $\sigma_S$ the prefix of the schedule $\sigma$ on machine $j_S$ that ends with $S$. The \emph{finish time} of $S$ in $\sigma$ is defined as 
\[
\ft\br{\sigma, S}=\sum_{S'\in \sigma_S} c\br{S', j_S}.
\]
For an element $u\in \cU$, let $S_u\ni u$ be the set $S\in\sigma$ minimizing $\ft\br{\sigma, S}$. We define \emph{covering time} of $u$ in $\sigma$ as $\ct\br{\sigma, u} = \ft\br{\sigma, S_u}$. The \emph{cost} of a schedule 
\[
\COST\br{\sigma}=\sum_{u\in \cU} w\br{u}\cdot \ct\br{\sigma, u}
\]
is the weighted sum of covering times of all of the elements in $\cU$. The goal of $\ProblemPMSSC$ is to find a feasible schedule $\sigma$ minimizing $\COST\br{\sigma}$. 
Let 
\[
C=\max_{j\in[m]}\brc{\sum_{S\in \cS}c\br{S, j}} 
\]
be a trivial upper bound on the length of any schedule.
We distinguish three cases of $\ProblemPMSSC$ according to the cost function:
\begin{enumerate}
    \item \textbf{Unit cost sets:} For every set $S\in\cS$, and $j\in[m]$, $c\br{S, j}=1$.
    \item \textbf{Identical machines:} For every set $S\in\cS$ and $j\in[m]$, $c\br{S, j}=c\br{S}$.
    \item \textbf{Unrelated machines:} No additional conditions on the cost function.
\end{enumerate}

In order to solve \textsc{PMSSC} efficiently we define the following subproblem called \textsc{Parallel Densest Subfamily} (\textsc{PDS}). 
For any assignment of sets to machines $\cF=\brc{\cF\br{j}\subseteq \cS}_{j=1}^m$, denote by 
\[
\Cmax\br{\cF}=\max_{j\in [m]}\brc{\sum_{S\in \cF\br{j}} c\br{S, j}} \]
the makespan of $\cF$.
For every $\cU'\subseteq \cU$, define its \emph{weight} by $w\br{\cU'}=\sum_{u\in \cU'} w\br{u}$. The \emph{density} of the assignment $\cF$ is then defined as 
\[
\Delta\br{\cF, \cU'} = \frac{w\br{\cov\br{\cF, \cU'}}}{\Cmax\br{\cF}},
\]
and we write $\Delta\br{\cF} = \Delta\br{\cF, \cU}$. Given a universe $\cU$, a family of covering subsets~$\cS$, $m$ machines, a cost function $c\colon \cS\times [m]\to \mathbb{N}$ and weights $w\colon \cU\to \mathbb{N}$, the problem is to find an assignment of subsets to machines $\cF$ maximizing $\Delta\br{\cF}$. Note, that for any maximization problem, we adopt a convention that an $\alpha$-approximation algorithm (with $\alpha\geq 1$) is an algorithm that returns a solution with value at least $\OPT/\alpha$, where $\OPT$ is the optimal value of the objective function.

Furthermore, the following \textsc{Maximum Coverage Multiple Knapsacks} (\textsc{MCMK}) problem will also be required: Given a universe $\cU$, a family of covering subsets~$\cS$, $m$ knapsacks, a cost function $c\colon \cS\times [m]\to \mathbb{N}$, a budget $B_j$ for every $j\in[m]$ and weights $w\colon \cU\to \mathbb{N}$, find an assignment of subsets to knapsacks $\cF$ such that for every knapsack $j\in [m]$ we have that $\sum_{S\in\cF\br{j}} c\br{S, j} \leq B_j$ and $w\br{\cov\br{\cF, \cU}}$ is maximized. We remark that here knapsacks serve the same purpose as the machines and the different naming is due to historical reasons.



\section{General scheme of approximating $\ProblemPMSSC$}
\label{sec:pmssc-scheme}

In this section we present a general $4\cdot\alpha$-approximation greedy algorithm for $\ProblemPMSSC$ using an $\alpha$-approximation algorithm for PDS. Then, we show how to obtain concrete approximation guarantees for $\ProblemPMSSC$ by plugging in specific algorithms for PDS.
The main algorithm is as follows:

\begin{algorithm}[H]
\caption{General scheme of approximating PMSSC}
\label{alg:pmssc_scheme}

\LinesNumbered

\textbf{Input:} A $\ProblemPMSSC$ instance $\br{\cU,\cS,m,c,w}$, and an $\alpha$-approximation algorithm $\cA$ for $\ProblemPDS$

Let $\cR\gets \cU$ and $\pi=\brc{\emptyset}_{j=1}^m$ be an empty schedule

\While{$\cR\neq \emptyset$}{

Run $\cA$ on PDS instance $\br{\cR,\cS, m, c, w}$ to obtain an assignment $\cF$

\For{$j\in [m]$}{

Append the sets in $\cF\br{j}$ to the end of $\pi\br{j}$

}

$\cR\gets \cR\setminus \cov\br{\cF, \cR}$

}

\Return $\pi$

\end{algorithm}

\begin{theorem}\label{theorem:PMSSC_approx}
    The above algorithm is a $4\cdot\alpha$-approximation for $\ProblemPMSSC$.
\end{theorem}
\begin{proof}
     We provide an adapted version of the histogram schrinkage argument used in various works on MSSC \cite{ApproximatingMinSumSetCover,AGeneralFrameworkForApproximatingMinSumOrderingProblems,Precedence-ConstrainedMinSumSetCover}. 
    Let $\sigma$ be an optimal schedule. Fix some iteration $i$ of the while loop.
    Denote by $\cF_i$ the assignment $\cF$ obtained using the $\alpha$-approximation algorithm for PDS in the $i$-th iteration.
     Moreover, let $\cR_i$ be the set of uncovered elements and $\cX_i$ be the set of elements covered for the first time in the $i$-th iteration, respectively. We have $\COST\br{\pi} \leq \sum_{i}w\br{\cR_i}\cdot \Cmax\br{\cF_i}$, since in the worst-case scenario $\cF_i$ contributes $\Cmax\br{\cF_i}$ to $\ct\br{\pi,u}$ for each $u\in \cR_i$.

    Let 
    \[
    \ell_i=\sum_{p=1}^{i}\Cmax\br{\cF_p}.
    \] 
    Order the elements of $\cU$ by $\ct\br{\pi, u}$. For each $u\in \cU$ we plot a rectangle of width $w\br{u}$ and height $\ell_i$, where $i$ is the index of the first $\cF_i$ which covered $u$. Clearly, the area under this histogram is equal to $\sum_{i}w\br{\cR_i}\cdot \Cmax\br{\cF_i}$, the upper bound on $\COST\br{\pi}$. By summing up over the values of the function, not over the elements (Lebesgue integral style), the area under the histogram can be partitioned into rectangles of width $w\br{\cR_i}$ and height $\Cmax\br{\cF_i}$. We also plot $\sigma$ similarily, ordering the elements of $\cU$ by $\ct\br{\sigma, u}$ and plotting a rectangle of width $w\br{u}$ and height $\ct\br{\sigma, u}$. Note that both plots are non-decreasing by definition.
    

\begin{figure}[!t]
\centering

\begin{minipage}[t]{0.48\linewidth}
\centering
\resizebox{\linewidth}{!}{%
\begin{tikzpicture}[xscale=1.2,yscale=0.8]
    \draw[thick] (0,0) -- (10,0) -- (10,8) -- (0,8) -- (0,0);
    
    \node[below=0.3cm] at (5, 0.2) {Weight};
    \node[left=0.3cm] at (0.2-0.3, 4) {Time};
    
    \node[above=0.5cm] at (5, 7.5) {\textbf{Greedy}};
    
    
    \filldraw[thick,fill=lightgray!60] (0, -0.3) rectangle (10, 2.0);
    
    \filldraw[thick,fill=lightgray!60] (1, 2.0) rectangle (10, 3.1);
    
    \filldraw[thick,fill=blue!40] (2, 3.1) rectangle (10, 4.7);

    \draw[thick,decorate,decoration={brace,amplitude=5pt}]
        (2,3.1) -- (2,4.7) node[midway,left=3pt] {$\Cmax\br{\cF_i}$};

    \draw[thick,decorate,decoration={brace,amplitude=5pt,mirror}]
        (2,-0.3) -- (2,4.7) node[midway,right=6pt] {$\ell_i$};

    \draw[thick,decorate,decoration={brace,amplitude=5pt}]
        (2,4.7) -- (4,4.7) node[midway,above=4pt] {$w\br{X_i}$};

    \draw[thick,decorate,decoration={brace,amplitude=5pt}]
        (2,3.1) -- (10,3.1) node[midway,above=4pt] {$w\br{R_i}$};
    
    \filldraw[thick,fill=lightgray!60] (4, 4.7) rectangle (10, 5.4);
    
    \filldraw[thick,fill=lightgray!60] (5, 5.4) rectangle (10, 7.2);
    
    \filldraw[thick,fill=lightgray!60] (7, 7.2) rectangle (10, 8);
\end{tikzpicture}
}
\end{minipage}\hfill%
\begin{minipage}[t]{0.48\linewidth}
\centering
\resizebox{\linewidth}{!}{%
\begin{tikzpicture}[xscale=1.2,yscale=0.8]
    \draw[thick] (0,0) -- (10,0) -- (10,8) -- (0,8) -- (0,0);
    
    \node[below=0.3cm] at (5, 0.2) {Weight};
    \node[left=0.3cm] at (0.2-0.3, 4) {Time};
    
    \node[above=0.5cm] at (5, 7.5) {\textbf{OPT}};
    
    \draw[thick]
        (0,0) -- (0,1.8) -- (2.0,1.8) -- (2.0,3.0) --
        (4.0,3.0) -- (4.0,4.1) -- (6.0,4.1) -- (6.0,5.0) --
        (7.2,5.0) -- (7.2,6.2) -- (8.4,6.2) -- (8.4,7.0) -- (10,7.0);

    \fill[lightgray!60] (5.0,0) rectangle (5.5,3.9);
    \draw[thick] (5.0,0) rectangle (5.5,3.9);

    \fill[lightgray!60] (5.5,0) rectangle (6.0,1.65);
    \draw[thick] (5.5,0) rectangle (6.0,1.65);

    \fill[blue!40] (6.0,0) rectangle (7.0,2.7);
    \draw[thick] (6.0,0) rectangle (7.0,2.7);

    \fill[lightgray!60] (7.0,0) rectangle (7.5,1.2);
    \draw[thick] (7.0,0) rectangle (7.5,1.2);

    \fill[lightgray!60] (7.5,0) rectangle (8.5,3.1);
    \draw[thick] (7.5,0) rectangle (8.5,3.1);

    \fill[lightgray!60] (8.5,0) rectangle (10.0,1.35);
    \draw[thick] (8.5,0) rectangle (10.0,1.35);

    \draw[thick,decorate,decoration={brace,amplitude=5pt}]
        (6.0,2.7) -- (7.0,2.7) node[midway,above=4pt] {$\frac{w\br{X_i}}{2}$};
    \draw[thick,decorate,decoration={brace,amplitude=5pt}]
        (6.0,0) -- (6.0,2.7) node[midway,left=3pt] {$h_i$};

    \draw[thick,decorate,decoration={brace,amplitude=5pt,mirror}]
        (10.0,0) -- (6.0,0) node[midway,above=4pt] {$\frac{w\br{R_i}}{2}$};
\end{tikzpicture}
}
\end{minipage}
\caption{Histogram of the greedy (left) and OPT with mapped columns (right).}
\label{fig:histogram-shrinkage}
\end{figure}

    For any iteration $i$ define 
    \[
    h_i=\frac{\Cmax\br{\cF_i}\cdot w\br{\cR_i}}{2\alpha\cdot w\br{\cX_i}}.
    \] 
    Now, map the $i$-th rectangle in the plot of $\pi$ to a column of height $h_i$ and width $w\br{\cX_i}/2$, placed between horizontal coordinates $w\br{\cR_{i+1}}/2$ and $w\br{\cR_i}/2$, measured from the right-hand end of the plot of $\sigma$. This gives a scaled histogram of $\pi$ with an area $4\cdot\alpha$ times smaller than its original histogram. For a visual depiction of plots of $\pi$, $\sigma$ and the mapping of $\pi$ unto the plot of $\sigma$ see Figure~\ref{fig:histogram-shrinkage}. 
    By $\pi_i$ we denote the collection of sets scheduled in $\pi$ up to the end of the $i$-th iteration of the while loop. 
    For any $t$, define $\sigma_t$ as the assignment obtained from $\sigma$ by removing every set $S$ with $\ft\br{\sigma,S}>t$.
\begin{lemma}
    For every $\cR_i\neq \emptyset$, we have $w\br{\cov\br{\sigma_{\fl{h_i}}, \cR_i}} \leq w\br{\cR_i}/2$.
\end{lemma}
\begin{proof}
    Denote by $\rho_i = \sigma_{\fl{h_i}}\setminus \pi_{i-1}$ the assignment obtained by removing all sets from $\pi_{i-1}$ from $\sigma_{\fl{h_i}}$.
    Assume that $\rho_i\neq\emptyset$, otherwise, $w\br{\cov\br{\sigma_{\fl{h_i}}, \cR_i}}=0\leq w\br{\cR_i}/2$. In the $i$-th iteration of the while loop, we run an $\alpha$-approximation algorithm for PDS with input $\br{\cR_i,\cS, m, c, w}$.
    Because $\rho_i$ is a feasible solution for PDS on the $\br{\cR_i,\cS, m, c, w}$ instance, we get that
\[
    \frac{w\br{\cov\br{\rho_i, \cR_i}}}{\Cmax\br{\rho_i}}\leq \alpha \cdot \frac{w\br{\cX_i}}{\Cmax\br{\cF_i}}.
\] 
Since $\cR_i\cap \cov\br{\pi_{i-1}}=\emptyset$ and 
\[
\Cmax\br{\rho_i}\leq \Cmax\br{\sigma_{\fl{h_i}}}=\max_{S\in \sigma_{\fl{h_i}}} \brc{\ft\br{\sigma, S}}\leq h_i,\]
we obtain
\[
w\br{\cov\br{\sigma_{\fl{h_i}}, \cR_i}} = w\br{\cov\br{\rho_i, \cR_i}}\leq \frac{\alpha\cdot h_i\cdot w\br{\cX_i}}{\Cmax\br{\cF_i}} = \frac{w\br{\cR_i}}{2}.
\]
\end{proof}

The above lemma shows that after $\fl{h_i}$ units of time in $\sigma$, at most $w\br{\cR_i}/2$ weight from $\cR_i$ is covered. Therefore, the plot of $\sigma$ has to have a height of at least $\fl{h_i}+1\geq h_i$ at the point $w\br{\cR_i}/2$. The top-left corner of the column corresponding to $h_i$ is at the point $\br{w\br{\cR_i}/2, h_i}$, which is below the plot of $\sigma$. Thus, since this plot of $\sigma$ is non-decreasing, the entire $i$-th column of height $h_i$ is below its plot. Hence, \[
\COST\br{\pi}\leq 4\cdot\alpha \cdot \sum_{i} h_i \cdot w\br{\cX_i}/2\leq 4\cdot\alpha \cdot \COST\br{\sigma}
\] 
as promised.
\end{proof}


\subsection{Approximating PDS via MCMK}
We give a reduction that applies to any machine model for which an approximation algorithm for the corresponding variant of \textsc{MCMK} is available.

\begin{algorithm}[H]
\caption{Parallel Densest Subfamily via MCMK}
\label{alg:pds_mcmk}

\LinesNumbered

\textbf{Input:} A PDS instance $\br{\cU,\cS,m,c,w}$, a constant $\delta>0$, and an $\alpha$-approximation algorithm $\mathcal{A}$ for the corresponding MCMK problem

Guess a budget $B$ by testing consecutive powers of $1+\delta$

Run $\mathcal{A}$ on the MCMK instance $\br{\cU, \cS, c, w, \brc{B_j=B}_{j\in[m]}}$ to obtain an assignment $\cF_B$

\Return $\cF=\argmax_B\brc{\Delta\br{\cF_B, \cU}}$

\end{algorithm}

\begin{theorem}\label{thm:pds-from-mcmk}
Let $\delta>0$. If there is a polynomial-time $\alpha$-approximation algorithm for \textsc{MCMK} under a given machine model, then Algorithm~\ref{alg:pds_mcmk} is a polynomial-time $\alpha\cdot(1+\delta)$-approximation algorithm for \textsc{PDS} under the same machine model.
\end{theorem}
\begin{proof}
The required number of guesses for $B$ is $\cO\br{\log_{1+\delta}C}$. Thus, the algorithm runs in polynomial time for every fixed $\delta>0$.
Let $\cF^*$ be an optimal \textsc{PDS} assignment and let $B^*=\Cmax\br{\cF^*}$. One of the tested values satisfies $
B^*\le B\le(1+\delta)\cdot B^*.$
For this value of $B$, the assignment $\cF^*$ is feasible for the MCMK instance because every knapsack has budget $B$. Therefore, if $\cF_B$ is returned by the $\alpha$-approximation algorithm, then 
\[
w\br{\cov\br{\cF_B,\cU}}\ge \frac{w\br{\cov\br{\cF^*,\cU}}}{\alpha}.
\]
Moreover, $\Cmax\br{\cF_B}\le B$, and hence,
\[
\Delta\br{\cF}\geq \Delta\br{\cF_B}
=\frac{w\br{\cov\br{\cF_B}}}{\Cmax\br{\cF_B}} 
\ge \frac{w\br{\cov\br{\cF^*}}}{\alpha \cdot B}
\ge \frac{\Delta\br{\cF^*}}{\alpha\cdot(1+\delta)}.
\]
\end{proof}

We will now use the following proposition to derive approximation algorithms for \textsc{PDS} from approximation algorithms for \textsc{MCMK}.
\begin{proposition}\label{proposition:mcmk-approx}
    For every $\epsilon>0$, there exists a polynomial-time $\br{\frac{e}{e-1}}^2+\epsilon$-approximation algorithm for MCMK on unrelated knapsacks.
\end{proposition}

\ifdefined\PMSSCProceedings
The proof and algorithm for Proposition~\ref{proposition:mcmk-approx} are given in the
full version (\href{https://arxiv.org/abs/2604.11388v5}{arXiv:2604.11388v5}).
\else
The proof of Proposition~\ref{proposition:mcmk-approx} together with the claimed procedure is deferred to Appendix~\ref{appendix:proof-mcmk-approx}.
\fi
Combining it with \cite{SimpleDeterministicApproximationForSubmodularMultipleKnapsackProblem} gives the following corollary.
\begin{corollary}\label{cor:pds-approximations}
For every constant $\epsilon>0$, there are polynomial-time approximation algorithms for \textsc{PDS} with the following approximation guarantees:
\begin{itemize}
\item $\frac{e}{e-1}+\epsilon$ on identical machines,
\item $\br{\frac{e}{e-1}}^2+\epsilon$ on unrelated machines.
\end{itemize}
\end{corollary}
\begin{proof}
Write $\gamma=e/(e-1)$. For identical machines, set 
\[
\delta_{\mathrm{id}}
=\frac{\epsilon}{\gamma\cdot \br{\gamma+1+\epsilon}}.
\]
and use the $1/(1-e^{-1}-\delta_{\mathrm{id}})$-approximation algorithm for MCMK from~\cite{SimpleDeterministicApproximationForSubmodularMultipleKnapsackProblem}. Theorem~\ref{thm:pds-from-mcmk} gives the approximation ratio of 
\[
\frac{1+\delta_{\mathrm{id}}}{1-e^{-1}-\delta_{\mathrm{id}}}
=\gamma+\epsilon.
\]

For unrelated machines, invoke Proposition~\ref{proposition:mcmk-approx} with accuracy parameter $\epsilon/2$, so the resulting MCMK approximation is $\alpha=\gamma^2+\epsilon/2$. Set 
\[
\delta_{\mathrm{un}}=\frac{\epsilon}{2\cdot\gamma^2+\epsilon}.
\]
Theorem~\ref{thm:pds-from-mcmk} then gives the approximation of 
\[
\alpha\cdot\br{1+\delta_{\mathrm{un}}}
=\br{\gamma^2+\frac{\epsilon}{2}}\cdot
\br{1+\frac{\epsilon}{2\cdot\gamma^2+\epsilon}}
=\gamma^2+\epsilon.
\]
\end{proof}

\subsection{Main Results}
Combining the above results with an appropriate choice of $\epsilon$, we get that.
\begin{corollary}
    There exist the following approximation algorithms for $\ProblemPMSSC$:
    \begin{itemize} 
    \item $\frac{4\cdot e}{e-1}+\epsilon< 6.33$ on identical machines,
    \item $4\cdot\br{\frac{e}{e-1}}^2+\epsilon< 10.02$ on unrelated machines.
    \end{itemize}
\end{corollary}

\section{Precedence constraints}\label{sec:precedence}
Now, consider a setup where $\cS$ is subject to precedence constraints given by a partially-ordered set $\cP=\br{\cS, \preceq}$. A \emph{precedence constraint} $S_1\preceq S_2$ means that $S_1$ needs to be completed before the start of $S_2$. Let $d$ be the length of the longest path in $\cP$. Let $\cP\br{S}$ be the \emph{closure} of $S\in\cS$ in $\cP$, i.~e., the set of all of predecessors of $S$ in $\cP$ (including~$S$). With precedence constraints, assignments specify start times. Finish times and the makespan $\Cmax\br{\cF}$ include idle time. We say that an assignment $\cF$ is \emph{precedence-closed} if for each $S\in \cF\br{j}$, all sets in $\cP\br{S}\setminus\brc{S}$ are scheduled before~$S$. For any $\cU'\subseteq \cU$, the \emph{density} of a precedence-closed assignment $\cF$ is defined as 
\[
\Delta\br{\cF, \cU'} = \frac{w\br{\cov\br{\cF, \cU'}}}{\Cmax\br{\cF}}.
\] 
and we write $\Delta\br{\cF}=\Delta\br{\cF, \cU}$. The \textsc{Precedence-Constrained Parallel Densest Subfamily} (\textsc{PCPDS}) problem is to find a precedence-closed assignment $\cF$ maximizing $\Delta\br{\cF}$.

The \emph{List-Scheduling} algorithm used for various parallel machine scheduling problems is as follows. It starts at the moment $t=0$ and repeatedly picks an available, yet unscheduled task and assigns it to any empty machine at the time, if one exists. If not, then it waits until some machine becomes available and repeats the process as long as there are unscheduled tasks. It is well-known that such an algorithm produces a schedule of length at most twice the optimum for $P|prec|C_{\max}$, see for example \cite{BoundsOnMultiprocessingTimingAnomalies}. In particular, we use a version, which schedules tasks according to the length of the critical path leading to a given task, which for unit costs is equivalent to scheduling tasks layer by layer. This means that all tasks at the same layer (i.e., having the same length of the critical path) are scheduled before moving on to the next layer.
\subsection{Unit cost sets}
To solve the problem for the case of unit cost sets, we significantly alter the algorithm of~\cite{Precedence-ConstrainedMinSumSetCover}.

\begin{algorithm}[H]
\caption{PCPDS algorithm for unit-cost instances}
\label{alg:pds:precedence}

\LinesNumbered

\textbf{Input:} A unit-cost \textsc{PCPDS} instance $\br{\cU,\cP,m,w}$

\For{$h\in [d]$}
{
    Obtain assignment $\cF_h$ by List-Scheduling all the sets up to depth $h$ in $\cP$
}
\For{$S\in \cS$}
{
    Obtain assignment $\cF_S$ by List-Scheduling the sets in $\cP\br{S}$
}
\Return $\cF=\argmax_{\cF\in\brc{\cF_h}_{h\in[d]}\cup\brc{\cF_S}_{S\in\cS}}\brc{\Delta\br{\cF, \cU}}$

\end{algorithm}

\begin{theorem}
The algorithm above returns a precedence-closed assignment $\cF$ such that $\Delta\br{\cF^*} \leq \cO\br{k^{2/3}}\cdot \Delta\br{\cF}$, where $\cF^*$ is an optimal precedence-closed assignment of sets to machines.
\end{theorem}
\begin{proof}
    Let 
    \[
    S=\argmax_{S\in \cS}\brc{\Delta\br{\cF_S}} \text{ and } \delta=\Delta\br{\cF_S}.\]
    We may assume $m\leq k$, since a schedule of $k$ sets needs at most $k$ machines.
    Let $\cF^*$ consist of sets $\brc{S_i}_{i=1}^{t}$. For each $S_i$, we have $\Delta\br{\cF_{S_i}} \leq \delta$. Let $c_i=\Cmax\br{\cF_{S_i}}$. Since we assign the sets in $\cP\br{S_i}$ to machines using List-Scheduling, for each $i\in [t]$, $c_i \leq 2\cdot \Cmax\br{\cF^*}$. Thus, 
    \[
        w\br{\cov\br{\cF^*}} \leq \sum_{i=1}^{t}w\br{\cov\br{\cP\br{S_i}}}\leq \delta\cdot \sum_{i=1}^{t} c_i \leq 2\delta\cdot t\cdot \Cmax\br{\cF^*}
    \]
            and 
    \[
        \frac{\Delta\br{\cF_S}}{\Delta\br{\cF^*}} = \frac{\delta\cdot \Cmax\br{\cF^*}}{w\br{\cov\br{\cF^*}}}\geq \frac{1}{2t}.
    \]

    Now, observe that trivially $\Delta\br{\cF^*}\leq m\cdot w\br{\cU}/t$. 
    We also have 
    \[
    \Delta\br{\cF_d}\geq \frac{w\br{\cU}}{k},
    \] 
    since $\cS$ covers $\cU$ and $\Cmax\br{\cF_d}\leq k$. Therefore, we get 
    \[
        \frac{\Delta\br{\cF_d}}{\Delta\br{\cF^*}}\geq \frac{t}{m\cdot k}.
    \]
    
    Denote $\cU_h=\cov\br{\cF_h}$. Let $h^*$ be the depth of the deepest $S\in \brc{S_i}_{i\in [t]}$ in~$\cP$. We have 
    \[
    \Delta\br{\cF^*}\leq \frac{w\br{\cU_{h^*}}}{h^*},
    \] 
    since $\cov\br{\cF^*}\subseteq \cU_{h^*}$ and $\Cmax\br{\cF^*}\geq h^*$. Each layer uses at most $\cl{k/m}$ unit time-slots. Thus, $\Cmax\br{\cF_{h^*}}\leq h^*\cl{k/m}\leq 2h^*k/m$. Hence, we get 
    \[
    \Delta\br{\cF_{h^*}}\geq \frac{m\cdot w\br{\cU_{h^*}}}{2h^*\cdot k},
    \] 
    implying that 
    \[
        \frac{\Delta\br{\cF_{h^*}}}{\Delta\br{\cF^*}}\geq \frac{m}{2k}.
    \]

    Using the AM-GM inequality we obtain the desired result
    \begin{align*}
        \max\brc{\frac{\Delta\br{\cF_S}}{\Delta\br{\cF^*}}, \frac{\Delta\br{\cF_{h^*}}}{\Delta\br{\cF^*}}, \frac{\Delta\br{\cF_d}}{\Delta\br{\cF^*}}} &\geq \max\brc{\frac{1}{2t}, \frac{t}{m\cdot k}, \frac{m}{2k}}\\&\geq \frac{1}{4^{1/3}\cdot k^{2/3}}.
    \end{align*}
\end{proof}

\subsection{Precedence constraints: out-forests, identical machines}
Out-forest precedence constraints are ones, where each set has at most one direct predecessor. In this case, we can improve the approximation to $\cO\br{\log k}$, even for a slightly more general setup of identical machines.

The \textsc{Group Steiner Orienteering} (\textsc{GSO}) problem consists of a symmetric metric $\br{V, d}$ rooted at $r\in V$, $g$ groups $\cX=\brc{X_p\subseteq V}_{p=1}^g$, a weight function $w\colon\cX\to \mathbb{N}$ on the groups, and a budget $K$. An $r$-tour is a sequence of vertices $\angl{r=u_0,u_1,\dots, u_t=r}$. The length of such a tour 
\[
\spr{\tau}=\sum_{i=1}^td\br{u_{i-1}, {u_i}}
\]
is the total length of all edges of the tour. Let $V\br{\tau}$ be the vertices visited by $\tau$.
For $\cX'\subseteq\cX$, define $\cov\br{\tau,\cX'}=\brc{X_p\in\cX'\colon X_p\cap V\br{\tau}\ne\emptyset}$ and $w\br{\cX'}=\sum_{X_p\in\cX'}w\br{X_p}$.
The goal of GSO is to maximize $w\br{\cov\br{\tau,\cX}}$ over $r$-tours of length at most $K$.
A bicriteria $\br{\alpha,\beta}$-approximation returns an $r$-tour $\tau$ with $\spr{\tau}\leq\alpha K$ and $w\br{\cov\br{\tau,\cX}}\geq\OPT/\beta$.
In \cite{ApproximationAlgorithmsForOptimalDecisionTreesAndAdaptiveTSPProblems} it is shown, that if $\br{V, d}$ is a tree metric, then there exists a polynomial time $\br{\cO\br{\log \spr{V}}, 4}$-bicriteria approximation algorithm for GSO. Using this, we showcase an $\cO\br{\log k}$-approximation algorithm for \textsc{PCPDS} with out-forest precedence constraints, which is as follows.

\begin{algorithm}[H]
\caption{PCPDS algorithm for identical machine, out-forest precedence constraints}
\label{alg:pds:precedence_outforest}

\LinesNumbered

\textbf{Input:} A \textsc{PCPDS} instance $\br{\cU,\cP,m,c,w}$, where $\cP$ is an out-forest

Construct an edge weighted rooted tree metric $T=\br{V, d}$ with a root $r$:  $V=\cS\cup\brc{r}$, for $S_1, S_2\in \cS$, $S_1S_2\in E\br{T}$ iff $S_1$ is the direct predecessor of $S_2$ in $\cP$.
For $S\in \cS$, if $S$ has no predecessor in $\cP$, add $rS$ to $E\br{T}$.
For $S\in \cS$, set the length of the edge connecting $S$ to its parent $P$ in $T$ to be $d\br{P, S}=c\br{S}$.

For each $u\in \cU$ create a group $X_u=\brc{S\in \cS\colon u \in S}$, s.t. $w\br{X_u}=w\br{u}$.

Guess a budget $B$ by testing consecutive powers of $2$.

Temporarily remove from $T$ all vertices $S\in \cS$ such that $d\br{r, S}>B$.

Run the $\br{\cO\br{\log \spr{V\br{T}}}, 4}$-approximation algorithm for \textsc{GSO} on the $\br{T, \cX, w, K=2m\cdot B}$ instance to obtain a tour $\tau_B$.

Replace each hop of $\tau_B$ by its tree path. Let $\cC_B$ be the sets visited by the resulting walk.

Use the List-Scheduling algorithm to schedule the sets in $\cC_B$ on $m$ machines to obtain an assignment $\cF_B$.

\Return $\cF=\argmax_{B}\brc{\Delta\br{\cF_B, \cU}}$.

\end{algorithm}
\begin{theorem}
    The above algorithm runs in polynomial time algorithm and returns a precedence-closed assignment $\cF$ such that $\Delta\br{\cF^*} \leq \cO\br{\log k}\cdot \Delta\br{\cF}$, where $\cF^*$ is an optimal precedence-closed assignment of sets to machines.
\end{theorem}
\begin{proof}
    Since the required number of guesses for $B$ is bounded by $\cO\br{\log C}$ the running time is $\text{poly}\br{n, m, \log C}$. Let $B^*=\Cmax\br{\cF^*}$. By guessing $B^*\leq B\leq2\cdot B^*$, we have that all sets used by $\cF^*$ are still present after step 4. Recall, that we run an $\br{\cO\br{\log \spr{V}}, 4}$-approximation algorithm for GSO with budget $K=2m\cdot B$. Consider the $r$-tour $\pi$ obtained by starting at $r$ and traversing the sets from $\cF^*$ in a depth-first-search manner in $T$ (possibly going back to $r$ several times). The length of this tour is at most $2 m\cdot B^*\leq K$, because the overall cost of all sets in $\cF^*$ is at most $m\cdot B^*$, and each such set has an edge of length $c\br{S}$ above it which needs to be traversed exactly two times. Also, no other edges of $T$ are traversed that way. Since each group corresponds to some element of $\cU$, given a budget $K$, for the optimal $r$-tour $\pi^*$ we have $w\br{\cov\br{\pi^*, \cX}}\geq w\br{\cov\br{\cF^*, \cU}}$. Moreover, $\cF_B$ contains all the sets of $\cC_B$ , visited by $\tau_B$, hence we get that 
    \[
    \begin{aligned}
    w\br{\cov\br{\cF_B, \cU}}
        &= w\br{\cov\br{\cC_B, \cU}} = w\br{\cov\br{\tau_B, \cX}} \\
        &\geq w\br{\cov\br{\pi^*, \cX}}/4 \geq \br{w\br{\cov\br{\cF^*}}}/4.
    \end{aligned}
    \] 
    Additionally, observe that for every set $S\in \cC_B$ and every $S'\in \cP\br{S}\setminus\brc{S}$, in order to reach $S$ in $T$, $\tau_B$ needs to visit $S'$ since $S'$ lies on the path from $r$ to $S$ in $T$. This means, that $\cC_B$ is precedence-closed. Thus, for each $B$, the assignment $\cF_B$, (and in particular $\cF$) is precedence-closed as well.

    Now, we reason about $\Cmax\br{\cF_B}$. Our analysis borrows from the well-known analysis of the List-Scheduling algorithm. Let $\angl{S_1,\dots,S_h}$ be a sequence of sets such that $S_1=\argmax_{S\in \cF_B}\brc{\ft\br{\cF_B, S}}$, $S_2$ is the unique direct predecessor of $S_1$ and so on, until there are no predecessors left. Let $t\in \left[0, \Cmax\br{\cF_B}\right)$. We show that always at least one of the two following conditions hold:
    \begin{enumerate}
        \item There is some set $S_d$ for $d\in[h]$, which is processed during the moment $t$.
        \item All of the machines are occupied at the moment $t$.
    \end{enumerate}
    If neither condition holds, the earliest unfinished job on the chain has all its predecessors completed.
    It is therefore available, and List-Scheduling would start it on an idle machine.

    Consider the interval $\cI_{\cF_B}=\left[0, \Cmax\br{\cF_B}\right)$. We split moments in $\cI_{\cF_B}$ into two types of intervals. Let $\cL_B^{Y}$ be the union of all moments $t$ when some $S_d$ is processed during $t\in \cI_{\cF_B}$. Then, $\spr{\cL_B^{Y}}\leq B$, since the length of the critical path in $T$, bounds the total time spent on these intervals. This in turn is at most $B$, because we discarded all sets further than $B$ from $r$. Let $\cL_B^{N}=\cI_{\cF_B}\setminus \cL_B^{Y}$ and let $P_B=\sum_{S\in\cC_B}c\br{S}$ be the total processing time. Then 
    \[
    \spr{\cL_B^{N}}\leq \frac{P_B}{m},
    \] since during those times all of the machines are occupied. Since the expanded closed walk crosses each edge of its subtree at least twice, $2P_B\leq\spr{\tau_B}=\cO\br{\log k}\cdot K$. Hence, 
    \begin{align*}
        \Cmax\br{\cF_B}& = \spr{\cL_B^{Y}} + \spr{\cL_B^{N}} \leq B+\frac{P_B}{m}\leq B+\frac{\spr{\tau_B}}{2m}= B+\frac{\cO\br{\log k}\cdot K}{2\cdot m} 
        \\&= B+\frac{\cO\br{\log k}\cdot 2m\cdot B}{2m}=B+\cO\br{\log k}\cdot B = \cO\br{\log k}\cdot B^*.
    \end{align*}
    Therefore, we get the following bound on the approximation ratio:
    \[
    \frac{\Delta\br{\cF^*}}{\Delta\br{\cF}} \leq \frac{w\br{\cov\br{\cF^*}}/\Cmax\br{\cF^*}}{w\br{\cov\br{\cF_B}}/\Cmax\br{\cF_B}} \leq \cO\br{\log k}.
    \]
\end{proof}

The analysis of the greedy algorithm for $\ProblemPMSSC$ from Section~\ref{sec:pmssc-scheme} extends to precedence-constrained instances with a small modification. At iteration $i$, we run the subroutine for \textsc{PCPDS} (instead of \textsc{PDS}) on the instance $\br{\cR_i, \cP, m, c, w}$, and the key observation is that $\rho_{i}$ is precedence-closed and therefore feasible for that subproblem. This leads to the following corollary:
\begin{corollary}
    There exist algorithms with the following approximation ratios for $\ProblemPMSSC$ subject to precedence constraints:
    \begin{itemize}
        \item $\cO\br{k^{2/3}}$, for unit cost sets and general precedence constraints.
        \item $\cO\br{\log k}$, for identical machines and out-forest precedence constraints.
    \end{itemize}
\end{corollary}

Recall that there is evidence that no algorithm with approximation ratio $\cO\br{k^{1/12-\epsilon}}$ nor $\cO\br{n^{1/6-\epsilon}}$ can exist even for $m=1$, assuming \textsc{Planted Dense Subgraph} conjecture.
For out-forest precedence constraints, we prove a nearly logarithmic lower bound in the universe size $n$.
\begin{theorem}
    For any $\epsilon > 0$, there is no $\cO\br{\log^{1-\epsilon} n}$-approximation algorithm for $\ProblemPMSSC$, subject to out-forest precedence constraints, even when $m=1$ and all weights are uniform, unless $NP\subseteq ZTIME\br{n^{\text{poly}\br{\log n}}}$.
\end{theorem}
\begin{proof}
    It suffices to consider $0<\epsilon<1$. We reduce from a problem similar to GSO, called the \textsc{Minimum Group Latency} (MGL). The input of MGL consists of a symmetric metric $\br{V, d}$, root $r\in V$ and $g$ groups $\cX=\brc{X_p\subseteq V}_{p=1}^g$. The latency $\lt\br{\tau,X_p}$ of a group $X_p$ is the length of the prefix of $\tau$ up to its first visit to a vertex in $X_p$. The goal is to find an $r$-tour $\tau$ that minimizes $\lt\br{\tau}=\sum_{p=1}^g\lt\br{\tau,X_p}$. We use tree metrics with positive rational edge lengths, which we scale to positive integers.
    We discard groups containing $r$, whose latency is zero.

    Let $\cS$ be the collection of sets, where we create a set $S_v$ for all $v\in V\setminus\brc{r}$. Since there is a one-to-one correspondence between $V\setminus\brc{r}$ and $\cS$, we might sometimes refer to a set $S_v$ simply by its corresponding vertex $v$ and vice-versa.
    Let $u$ be the parent of $v$ in the original tree metric. Then, let $c\br{S_v}=d\br{u, v}$ and unless $u=r$, set $S_u\preceq S_v$, thus ensuring that $\cP$ forms an out-forest. Let $\cU=\cX$ and for $X_p\in \cX$, $X_p\in S_v$ iff $v\in X_p$, for all $v\in V\setminus\brc{r}$. All weights are uniform, and from now on we ignore them.

    Given an optimal MGL tour $\tau^*$, expand it along the tree edges and schedule the jobs in order of first visit. This order respects precedence. By the time the walk first visits a vertex, it has already traversed the incoming edge of each job scheduled up to that vertex. Each group is therefore covered no later in this schedule than in the tour. For an optimal schedule $\sigma^*$, this gives $\COST\br{\sigma^*}\leq\lt\br{\tau^*}$.

    Conversely, let $\sigma$ be an $\alpha$-approximate schedule. For each $i=0,1,\dots$ until $2^i$ reaches the last completion time, traverse the subtree of jobs completed by $2^i$, starting and ending at $r$. This subtree is rooted and its edge lengths sum to at most $2^i$, since completed jobs form a precedence-closed prefix on one machine.
    Concatenating the traversals visits those jobs by time $2\sum_{j=0}^i2^j<4\cdot2^i$. Thus a group covered in $\sigma$ at time $t\geq1$ is visited by time $8t$, taking the least $i$ with $2^i\geq t$. The resulting tour has latency at most $8\COST\br{\sigma}\leq8\alpha\lt\br{\tau^*}$.

    The reduction creates one universe element for each group, so $n=g$. An $\cO\br{\log^{1-\epsilon} n}$-approximation for PMSSC would therefore give an $\cO\br{\log^{1-\epsilon}g}$-approximation for MGL, which, under the stated complexity assumption, contradicts the $\cO\br{\log^{1-\epsilon}g}$ inapproximability result for MGL of Chakrabarty and Swamy~\cite[Theorem~2.5]{FacilityLocationWithClientLatenciesLPBasedTechniquesForMinimumLatencyProblems}.
\end{proof}

\section{Conclusions and open problems}
We have provided a series of approximation algorithms for PMSSC. The ratios are $\frac{4\cdot e}{e-1}+\epsilon<6.33$, $4\cdot \br{\frac{e}{e-1}}^2+\epsilon<10.02$, $\cO\br{k^{2/3}}$, and $\cO\br{\log k}$ for identical machines, unrelated machines, unit-cost sets with precedence constraints, and identical machines with out-forest precedence constraints, respectively.
These are the first results for PMSSC, to the best of our knowledge.
As a part of our work, we have also provided an $\br{\frac{e}{e-1}}^2+\epsilon$-approximation algorithm for the \textsc{MCMK} on unrelated machines using a new configuration-LP formulation. We believe this approach could find applications in other cover-like problems.

It would be of interest to incorporate different scheduling constraints into $\ProblemPMSSC$, such as release dates $\cS$, or due dates on $\cU$. We remark, that all of our algorithms can be adapted to work for release dates, with a factor $2$ loss in the approximation ratio. Upon deciding on a budget $B$ (if the algorithm does not have this step, introduce it), temporarily remove all of the sets with release dates above $B$. Then, ignore the release dates, obtain an assignment as before, and modify it by adding idle time of length $B$ at the beginning of the assignment.


\begin{credits}
\subsubsection{\ackname}
Michał Szyfelbein has been partially supported by National Science Centre (Poland) grant number 2025/59/B/ST6/01072.

\subsubsection{\discintname}
The author declares no competing interests relevant to this work.
\end{credits}

\bibliographystyle{lncs-template/llncs/splncs04}
\bibliography{lipics-v2021-sample-article}

\ifdefined\PMSSCWithAppendix
  \clearpage
  \section*{Appendix}
  \appendix
  \newcommand{\PMSSCGroupBudgetReference}{%
    \href{https://doi.org/10.1007/978-3-540-27821-4_7}{Chekuri and Kumar (2004)}%
  }
\section{Proof of Proposition~\ref{proposition:mcmk-approx}: MCMK on Unrelated Knapsacks}
\label{appendix:proof-mcmk-approx}

In this section we provide an LP-based approximation algorithm for MCMK on unrelated knapsacks, i.e., each set can have a different cost of being assigned to each knapsack. To this end, we restrict ourselves to a setup, where every set can be assigned only to one, prespecified knapsack, and thus has only one associated cost. This is because we can replace each set $S\in \cS$ from the original instance by its $m$ copies $\brc{S_j}_{j\in[m]}$, each allowed to be assigned only to knapsack $j$, with $c\br{S_j}=c\br{S, j}$. Without loss of generality, we might assume that at most one representative of each original set is assigned to some knapsack. Otherwise, we could simply remove the extra copies without affecting feasibility or coverage. Therefore, all sensible solutions in both instances coincide and the instances are equivalent.
For each $S\in\cS$, let $j_S\in[m]$ denote the unique knapsack to which it is assigned, and let $
\cS_j=\brc{S\in\cS\colon j_S=j}$ be the collection of sets which can be assigned to knapsack $j\in[m]$.

Recall, that our goal is to get an $\br{\frac{e}{e-1}}^2+\epsilon$ approximation algorithm for MCMK on unrelated knapsacks.
A slightly worse bound of $ \frac{2\cdot e}{e-1}$ can be obtained using the ideas of \cite{MaximumCoverageWithClusterConstraints}, since their LP-based approach can be adapted to our setting as well. Indeed, after the same preprocessing as described above, their LP can directly encode the constraints that a given set can only be assigned to one, prespecified knapsack. In fact the authors use a similar restriction in order to prevent fractionally assigning very costly sets to small knapsacks. The rounding subroutine can then be adapted accordingly.

\providecommand{\PMSSCGroupBudgetReference}{\cite{MaximumCoverageWithGroupBudgetConstraints}}
There is also a way to obtain an $\frac{e}{e-1}+1$-approximaton for uniform weights using the approach of \PMSSCGroupBudgetReference{}, which is slightly worse than our target bound (but better then $\frac{2\cdot e}{e-1}$). In this work the authors regard a variant of the \textsc{Budgeted Maximum Coverage}, where there are multiple groups of sets (of possibly exponential size), each with a budget constraint on how many sets can be selected from that group. There is also a common bound denoting the overall amount of sets that can be assigned and the goal is to maximize the coverage. The authors give a $\alpha+1$-approximation algorithm for this problem, where $\alpha$ is the approximation ratio of an oracle that finds the currently best possible selection of sets from a single group. In our case, this choice can be made to encode all feasible configurations of sets assigned (or not) to each knapsack. The amount of available configurations for each group might be exponential, however finding the best selection among them is an instance of the classical \textsc{Budgeted Maximum Coverage}, and hence can be found $\alpha$-approximately for $\alpha = \frac{e}{e-1}$ in polynomial time. Allowing the choice of only one set per group, representing the chosen configuration per each knapsack and setting the global bound to $m$ corresponds to finding a solution for MCMK.

We omit the details of these approaches since they are relatively straightforward, and more importantly, our algorithm provides a better approximation ratio. 
To the best of our knowledge, these are the first explicit non-trivial approximation guarantees for this setup.

\subsection{A configuration LP for MCMK}
\label{sec:appendix-configuration LP}
Although our goal is to design an entirely new approximation algorithm for MCMK on unrelated knapsacks, the previously discussed reduction to the group-constrained variant of the \textsc{Budgeted Maximum Coverage} is a source of intuition behind the configuration LP, which we utilize. Consider what would happen if we were to explicitly enumerate all possible assignments of sets to knapsacks directly, and then pick one such configuration per knapsack while maximizing the coverage. Of course, this approach is computationally infeasible due to potentially exponential number of possibilities. However, we can encode such a choice into an exponential-size configuration LP, which allows us to reason about all such possibilities implicitly. By analyzing the dual of this LP, we show that it is possible to solve it approximately in polynomial time. Importantly, during this computation, we also obtain a polynomial-size support 
of the optimal primal solution. This in turn allows us to approximately solve the primal itself and round it to obtain an integral solution of high quality.

Here, we also remark that the correspondence between our algorithm and the group-constrained variant of the \textsc{Budgeted Maximum Coverage} has another aspect. Whereby their algorithm, adapted to our setting would behave greedily with respect to the oracle that selects the best configuration for each group, our algorithm uses the \textsc{Budgeted Maximum Coverage} as a separation oracle to identify violated constraints in the dual of the configuration LP. Despite these differences, the solutions are conceptually related, since in both cases the point is to reduce the multi-knapsack problem to a series of single-knapsack subproblems that can be handled using existing techniques.

In order to make these ideas more precise, for each knapsack $j\in[m]$, let
\[
\cC_j=\brc{C\subseteq\cS_j\colon \sum_{S\in C}c\br{S}\le B_j}
\]
be the family of its feasible \emph{configurations}. It is easy to see, that for each knapsack $j\in [m]$, these are the only feasible assignment choices that respect the budget constraint. Naturally, extend the definition of coverage to configurations by setting $\cov\br{C}=\bigcup_{S\in C}\cov\br{S}$ for each $C\in\cC_j$ and cost by setting $c\br{C}=\sum_{S\in C}c\br{S}$. Let $\cK=\brc{C_j}_{j\in[m]}$ be a \emph{configuration collection}. Then, we define its coverage as $\cov\br{\cK}=\bigcup_{j\in[m]}\cov\br{C_j}$. Our goal becomes to find a configuration collection $\cK$ 
maximizing $w\br{\cov\br{\cK}}$.

For $C\in\cC_j$, let $\lambda_{j,C}$ be a variable that indicates the extent to which configuration $C$ is selected for $j$ and for $u\in\cU$, let $y_u$ be a variable that indicates the extent to which element $u$ is covered by the entire solution. The following configuration LP has exponentially many variables, but the number of constraints is polynomial in the input size (not counting the domain of the variables $\lambda_{j,C}$):
\begin{align}
\max\quad & \sum_{u\in\cU}w\br{u}\cdot y_u \label{eq:config-primal-obj} \\
\text{s.t.}\quad
& \sum_{C\in\cC_j}\lambda_{j,C}=1
&& \forall j\in[m], \label{eq:config-primal-knapsack} \\
& y_u\le \sum_{j\in[m]}\sum_{C\in\cC_j, u\in\cov\br{C}}\lambda_{j,C}
&& \forall u\in\cU, \label{eq:config-primal-coverage} \\
& y_u\le 1
&& \forall u\in\cU, \label{eq:config-primal-unit} \\
& y_u\ge 0
&& \forall u\in\cU, \label{eq:config-primal-y} \\
& \lambda_{j,C}\ge 0
&& \forall j\in[m],\ \forall C\in\cC_j. \label{eq:config-primal-lambda}
\end{align}

The constraint~\eqref{eq:config-primal-knapsack} ensures that exactly one configuration is selected for each knapsack.
The constraint~\eqref{eq:config-primal-coverage} ensures that the coverage variable $y_u$ for each element $u\in\cU$ does not exceed the total extent to which $u$ is covered by the selected configurations across all knapsacks. The constraint~\eqref{eq:config-primal-unit} ensures that $y_u$ does not exceed 1, reflecting that an element cannot be covered more than fully. The remaining constraints enforce non-negativity. The objective is to maximize the total weighted coverage of the elements.
Denote by $\OPT_{\mathrm{CLP}}$ the value of the configuration LP. Observe, that for any choice of $\lambda$, the best possible value of $y$ is 
\[
y_u=\min\brc{1, \sum_{j\in[m]}\sum_{C\in\cC_j, u\in\cov\br{C}}\lambda_{j,C}} \quad \forall u\in\cU.
\] 
Therefore, from now on we treat 
$y$ as implicitly defined by $\lambda$, and the configuration LP can be equivalently written as a maximization over $\lambda$ only.
\begin{lemma}\label{lem:config-relaxation}
We have $\OPT_{\mathrm{CLP}} \ge \OPT$.
\end{lemma}
\begin{proof}
Fix an optimal configuration collection $\cK^*=\brc{C_j}_{j\in[m]}$. Since for each $j\in[m]$, $C_j\in\cC_j$, set $\lambda_{j,C_j}=1$, and $\lambda_{j,C}=0$ otherwise. Also, set $y_u=1$ if $u\in\cov\br{\cK^*}$, and $y_u=0$ otherwise. This solution satisfies constraints~\eqref{eq:config-primal-knapsack}-\eqref{eq:config-primal-lambda}, and its objective value equals $w\br{\cov\br{\cK^*}}$. Therefore, the optimal value of the configuration LP can only be larger or equal to the optimal value of the integral solution.
\end{proof}

We next describe how to solve the configuration LP approximately without enumerating its variables. To this end, we write the dual of our LP. It contains a nonnegative variable $a_u$ for constraint~\eqref{eq:config-primal-coverage}, a nonnegative variable $b_u$ for constraint~\eqref{eq:config-primal-unit}, and a free variable $p_j$ for constraint~\eqref{eq:config-primal-knapsack}:
\begin{align}
\min\quad & \sum_{j\in[m]}p_j+\sum_{u\in\cU}b_u \label{eq:config-dual-obj} \\
\text{s.t.}\quad
& a_u+b_u\ge w\br{u}
&& \forall u\in\cU, \label{eq:config-dual-element} \\
& p_j\ge \sum_{u\in\cov\br{C}}a_u
&& \forall j\in[m],\ \forall C\in\cC_j, \label{eq:config-dual-configuration} \\
& a_u,b_u\ge 0
&& \forall u\in\cU. \label{eq:config-dual-nonnegative}
\end{align}
Note, that since we include the empty configuration in $\cC_j$, the constraint associated with it implies that $p_j\ge0$, even though $p_j$ is formally unrestricted.

Now, for convenience, we write $\gamma=\frac{e}{e-1}$.

\begin{lemma}\label{lem:config-approx-separation}
There is a polynomial-time weak separation oracle for the dual with the following guarantee: whenever the oracle accepts a point $(p,a,b)$, the point $(\gamma\cdot p,a,b)$ is feasible for the dual.
\end{lemma}
\begin{proof}
Constraints~\eqref{eq:config-dual-element} and~\eqref{eq:config-dual-nonnegative} can be checked directly. Fix a knapsack $j\in[m]$ and consider the optimization problem
\[
\max_{C\in\cC_j}\brc{\sum_{u\in\cov\br{C}}a_u}.
\]
The key observation is that this is an instance of the \textsc{Budgeted Maximum Coverage}. Indeed, for each knapsack $j\in[m]$, in order to find the constraint which maximizes the right-hand side of~\eqref{eq:config-dual-configuration}, we need to select a configuration $C\in\cC_j$, which is a subset of $\cS_j$, such that $c\br{C}\leq B_j$ and the sum of the values $a_u$ for the covered elements is maximized.
Apply its polynomial-time $\gamma$-approximation algorithm of~\cite{TheBudgetedMaximumCoverageProblem}, and let $\widehat{C}_j$ be the returned configuration. If
\[
\sum_{u\in\cov\br{\widehat{C}_j}}a_u>p_j,
\]
then the constraint~\eqref{eq:config-dual-configuration} for $\widehat{C}_j$ is violated, and the oracle returns it. Otherwise, if the oracle accepts, then
\begin{align*}
\frac{1}{\gamma}\cdot\max_{C\in\cC_j}\sum_{u\in\cov\br{C}}a_u
&\le \sum_{u\in\cov\br{\widehat{C}_j}}a_u
\le p_j.
\end{align*}
Consequently,
\[
\sum_{u\in\cov\br{C}}a_u\le \gamma\cdot p_j
\qquad \forall C\in\cC_j.
\]
If this holds for every knapsack $j\in [m]$, then the point $(\gamma\cdot p,a,b)$ satisfies every constraint of the dual.
\end{proof}

Now, using our weak separation oracle, we run the ellipsoid method on the dual. Moreover, we record every configuration returned by the oracle, and together with the empty configurations, upon termination of the ellipsoid method for the dual, we solve the primal restricted only to those configurations. For a solution $\lambda$, define
\[
\supp\br{\lambda}=\brc{(j,C)\colon j\in[m],\ C\in\hat{\cC}_j,\ \lambda_{j,C}>0}.
\]
to be its support. The point is that, solving the dual using the ellipsoid method certifies that there exists a $\gamma$-approximate solution $\lambda$ with polynomially bounded support for the restricted primal LP. This is formalized in the following lemma.  
\begin{lemma}\label{lem:config-approx-solution}
For every constant $\delta>0$, in polynomial time we can compute a feasible configuration LP solution $\lambda$ with polynomially bounded $\spr{\supp\br{\lambda}}$ and value $V$ such that
\[
V\ge \frac{\OPT_{\mathrm{CLP}}}{\gamma+\delta}.
\]
\end{lemma}
\begin{proof}
For each knapsack $j\in[m]$, the ellipsoid method either finds a violated constraint of the dual or accepts the current point. Let $\hat{\cC}_j$ be the set of configurations returned by the oracle for knapsack $j$.
We also add the empty configuration $C_j=\emptyset$ to every $\hat{\cC}_j$. Since for any $\delta>0$, the ellipsoid method makes polynomially many separation oracle calls and every call returns at most one configuration, the quantity $\sum_{j\in[m]}|\hat{\cC}_j|$ is polynomially bounded.

After the ellipsoid method terminates, we solve the following restricted primal LP:
\begin{align*}
\max\quad & \sum_{u\in\cU}w\br{u}\cdot y_u \\
\text{s.t.}\quad
& \sum_{C\in\hat{\cC}_j}\lambda_{j,C}=1
&& \forall j\in[m], \\
& y_u\le \sum_{j\in[m]}\sum_{C\in\hat{\cC}_j, u\in\cov\br{C}}\lambda_{j,C}
&& \forall u\in\cU, \\
& 0\le y_u\le 1
&& \forall u\in\cU, \\
& \lambda_{j,C}\ge0
&& \forall j\in[m],\ \forall C\in\hat{\cC}_j.
\end{align*}
The only difference between the full configuration LP and the restricted LP is that the restricted LP only includes variables $\lambda_{j,C}$ for configurations $C\in\hat{\cC}_j$, for $j\in [m]$. In particular, it has polynomial size and thus, can be solved exactly in polynomial time. 

We extend a solution of the restricted LP to all configurations by (implicitly) setting $\lambda_{j,C}=0$ for every $C\in\cC_j\setminus\hat{\cC}_j$. Under this extension, every constraint of the full configuration LP is identical to the corresponding constraint of the restricted LP. Hence, the extended solution is feasible for the full configuration LP and we have
\[
|\supp\br{\lambda}|\le \sum_{j\in[m]}|\hat{\cC}_j|,
\]
which is polynomially bounded. 

It remains to prove the claimed bound on the value of this solution. For a threshold $T\ge0$, consider the feasibility problem obtained by adding
\[
\sum_{j\in[m]}p_j+\sum_{u\in\cU}b_u\le T
\]
to the constraints of the dual. Suppose that the weak separation oracle accepts a point $(p,a,b)$ satisfying this inequality. By Lemma~\ref{lem:config-approx-separation}, $(\gamma\cdot p,a,b)$ is feasible for the full dual, and its value is at most
\[
\sum_{j\in[m]}\gamma\cdot p_j+\sum_{u\in\cU}b_u
\le \gamma\cdot\br{\sum_{j\in[m]}p_j+\sum_{u\in\cU}b_u}
\le \gamma \cdot T,
\]
where the first inequality uses $b_u\ge0$ and $\gamma\ge1$. Weak duality therefore gives
\[
\OPT_{\mathrm{CLP}}\le\gamma \cdot T.
\]
Consequently, if $T<\OPT_{\mathrm{CLP}}/\gamma$, the oracle cannot accept any point and the ellipsoid method certifies that the threshold feasibility problem is infeasible.

The ellipsoid optimization performs a binary search over $T$. If $\OPT_{\mathrm{CLP}}=0$, the all-zero coverage solution already proves the lemma. Otherwise, use relative precision $1+\rho$, where $\rho=\delta/\gamma$. The search then in polynomial time produces a rejected threshold $\widehat{T}$ satisfying
\[
\widehat{T}\ge
\frac{
\OPT_{\mathrm{CLP}}/\gamma
}{1+\rho}
=\frac{\OPT_{\mathrm{CLP}}}{\gamma+\delta}.
\]

Consider the ellipsoid run for this rejected threshold $\widehat{T}$. Its infeasibility certificate uses the explicitly checkable dual constraints and only configuration constraints returned by the oracle. Hence, after collecting these configurations in the families $\hat{\cC}_j$, the system
\begin{align*}
&a_u+b_u\ge w\br{u}
&& \forall u\in\cU, \\
&p_j\ge\sum_{u\in\cov\br{C}}a_u
&& \forall j\in[m],\ \forall C\in\hat{\cC}_j, \\
&a_u,b_u\ge0
&& \forall u\in\cU, \\
&\sum_{j\in[m]}p_j+\sum_{u\in\cU}b_u\le\widehat{T}
\end{align*}
is infeasible. Equivalently, the restricted dual, obtained by retaining only the configuration constraints indexed by $C\in\hat{\cC}_j$, has optimum strictly larger than $\widehat{T}$. This restricted dual is the dual of the restricted primal displayed above. By strong LP duality, the optimal value $V$ of that restricted primal satisfies
\[
V>\widehat{T}\ge\frac{\OPT_{\mathrm{CLP}}}{\gamma+\delta}.
\]
Together with the extension argument above, this proves all claims of the lemma.
\end{proof}

We now give a generic rounding procedure recovering the configuration of each knapsack $j\in[m]$ from the LP solution. Independently, for each knapsack $j\in [m]$, choose one configuration $C_j\in\cC_j$ with probability $\lambda_{j,C_j}$. Let the resulting configuration collection be $\cK=\{C_j\}_{j\in[m]}$. Note, that this sampling runs in polynomial-time because only configurations in $\supp\br{\lambda}$ have non-zero probability of being chosen.

\begin{lemma}\label{lem:config-rounding-feasible}
$\cK$ is a feasible configuration collection.
\end{lemma}
\begin{proof}
For every $j\in[m]$, the selected configuration $C_j$ belongs to $\cC_j$, and hence
\[
c\br{C_j}=\sum_{S\in C_j}c\br{S}\le B_j.
\]
\end{proof}

\begin{lemma}\label{lem:config-rounding-value}
We have $\E{w\br{\cov\br{\cK}}}\ge V/\gamma$.
\end{lemma}
\begin{proof}
For every $j\in[m]$ and $u\in\cU$, let
\[
p_{j,u}=\sum_{C\in\cC_j, u\in\cov\br{C}}\lambda_{j,C}.
\]
be the probability that the configuration selected for knapsack $j$ covers $u$. Since all of the configurations are sampled independently for different knapsacks,
\[
\Pr[u\in \cov\br{\cK}]=1-\prod_{j\in[m]}\br{1-p_{j,u}},
\]
where we note that this quantity can be calculated efficiently, which will come useful later.

Using $1-z\le e^{-z}$ and constraint~\eqref{eq:config-primal-coverage}, we obtain
\begin{align*}
\Pr[u\in \cov\br{\cK}]
&\ge 1-\exp\br{-\sum_{j\in[m]}p_{j,u}} \\
&\ge \frac{1}{\gamma}\cdot\min\brc{1,\sum_{j\in[m]}p_{j,u}}
\ge \frac{y_u}{\gamma}.
\end{align*}

Therefore, by linearity of expectation,
\begin{align*}
\E{w\br{\cov\br{\cK}}}=
\sum_{u\in\cU}w\br{u}\cdot\Pr[u\in \cov\br{\cK}] \ge \frac{1}{\gamma}\cdot\sum_{u\in\cU}w\br{u}\cdot y_u
=\frac{V}{\gamma}.
\end{align*}
\end{proof}

Although the expected value of the coverage is at least $V/\gamma$, we have no guarantee that the randomized rounding will always achieve this value. However, a relatively straightforward application of the conditional expectations method allows us to derandomize this process efficiently. 
For $t\in [m]$, let $\cK_t=\bigcup_{j=1}^{t} C_j$ and $\cK_0=\emptyset$. Suppose configurations $\cK_{t-1}$ has already been fixed for some $t\in[m]$.
For every configuration $C\in \supp\br{\lambda_{t,\cdot}}$, let $\cK_C=\cK_{t-1}\cup\{C\}$ and let
\[
\Phi_t(\mathcal K_{t-1},C)
=\sum_{u\in \cov\br{\cK_C}} w\br{u}+\sum_{u\in \cU\setminus\cov\br{\cK_C}} w\br{u}\cdot\Pr[u\in \cov\br{\cK}].
\]
be the expected value of the coverage conditioned on the event that the configurations for $j\in [t-1]$ are $\cK_{t-1}$ and on choosing $C$ for knapsack $t$.

\begin{lemma}\label{lem:config-conditional-expectations}
In polynomial time, the randomized rounding can be derandomized without decreasing its expected guarantee.
\end{lemma}
\begin{proof}
Before fixing the configuration of knapsack $t$, the conditional expected value is
\[
\mathbb{E}_{C\sim\lambda_{t,\cdot}}[{\Phi_t\br{\cK_{t-1}, C}}]
=\sum_{C\in\cC_t}\lambda_{t,C}\cdot\Phi_t\br{\cK_{t-1}, C}.
\]
This is a convex combination of the values $\Phi_t\br{\cK_{t-1}, C}$. Hence, some configuration $C_{t}\in\supp\br{\lambda_{t,\cdot}}$ satisfies
\[
\Phi_t\br{\cK_{t-1}, C_{t}}\ge
\sum_{C\in\cC_t}\lambda_{t,C}\cdot\Phi_t\br{\cK_{t-1}, C}.
\]
Therefore, we can choose $C_t=\argmax_{C\in\supp\br{\lambda_{t,\cdot}}}\brc{\Phi_t\br{\cK_{t-1}, C}}$. Since this choice does not decrease the conditional expected value throughout the process, repeating the argument for $t\in[m]$ results in a deterministic configuration collection $\cK=\brc{C_j}_{j\in[m]}$, such that $w\br{\cov\br{\cK}}\ge V/\gamma$.
\end{proof}

\begin{proof}[Proof of Proposition~\ref{proposition:mcmk-approx}]
Fix $\epsilon>0$ and apply Lemma~\ref{lem:config-approx-solution} with $
\delta=\frac{\epsilon}{\gamma}.$
By Lemmas~\ref{lem:config-relaxation} and~\ref{lem:config-approx-solution}, the recovered configuration LP solution has value
\[
V\ge\frac{\OPT_{\mathrm{CLP}}}{\gamma+\delta}
\ge\frac{\OPT}{\gamma+\delta}.
\]
Lemma~\ref{lem:config-conditional-expectations} produces deterministically a set of configurations $\cK$ with $w\br{\cov\br{\cK}}\ge V/\gamma$. Therefore
\[
w\br{\cov\br{\cK}}\geq \frac{V}{\gamma}
\ge \frac{\OPT}{\gamma\cdot\br{\gamma+\delta}}
=\frac{\OPT}{\gamma^2+\epsilon},
\]
where the equality follows from $\delta=\epsilon/\gamma$. Thus, the algorithm is a $(\gamma^2+\epsilon)$-approximation. Finally, Lemma~\ref{lem:config-rounding-feasible} shows that the resulting solution is feasible.
\end{proof}

\fi

\end{document}